\documentclass[a4paper,10pt]{article}
\usepackage[english]{babel}

\usepackage[utf8]{inputenc}     % Encoding of the file
\usepackage{a4wide}
\usepackage{amsfonts}
\usepackage{amsthm}
\usepackage{amssymb}
\usepackage{amsmath}
\usepackage{mathtools}
\usepackage{tikz}
\usetikzlibrary{arrows,automata,shapes,shapes.geometric}
\usepackage{hyperref}
\usepackage{xcolor}

\newtheorem{theorem}{Theorem}

\newtheorem{lemma}{Lemma}
\newtheorem{proposition}{Proposition}
\newtheorem{corollary}{Corollary}

\newcommand{\trans}[1]{\mathchoice{\xrightarrow{#1}}{\xrightarrow{\smash{\lower1pt\hbox{$\scriptstyle #1$}}}}{\xrightarrow{#1}}{\xrightarrow{#1}}}
\newcommand{\B}[0]{\mathbb{B}}

\newcommand{\N}{\mathbb{N}}

\newcommand{\nhl}[0]{h_\ell}
\newcommand{\onh}[0]{\overline{h}}
\newcommand{\onhl}[0]{\overline{h}_\ell}
\newcommand{\unh}[0]{\underline{h}}
\newcommand{\unhl}[0]{\underline{h}_\ell}

\newcommand{\occ}[2]{|#1|_{#2}}

\newcommand{\entropyprob}[0]{\mathfrak{h}}

\newcommand{\ohc}[0]{\overline{hc}}
\newcommand{\uhc}[0]{\underline{hc}}

\newcommand{\betar}[0]{\beta}
\newcommand{\betarl}[0]{\beta_\ell}
\newcommand{\obetar}[0]{\overline{\beta}}
\newcommand{\obetarl}[0]{\overline{\beta_\ell}}
\newcommand{\ubetar}[0]{\underline{\beta}}
\newcommand{\ubetarl}[0]{\underline{\beta_\ell}}
\newcommand{\gammar}[0]{\gamma}
\newcommand{\gammarl}[0]{\gamma_\ell}

\newcommand{\ogammar}[0]{\overline{\gamma}}
\newcommand{\ogammarl}[0]{\overline{\gamma_\ell}}
\newcommand{\ugammar}[0]{\underline{\gamma}}
\newcommand{\ugammarl}[0]{\underline{\gamma_\ell}}

\DeclareUnicodeCharacter{D7}{\times}            % ×
\DeclareUnicodeCharacter{393}{\Gamma}           % Γ
\DeclareUnicodeCharacter{394}{\Delta}           % Δ
\DeclareUnicodeCharacter{39B}{\Lambda}          % Λ
\DeclareUnicodeCharacter{3A0}{\Pi}              % Π
\DeclareUnicodeCharacter{3B1}{\alpha}           % α
\DeclareUnicodeCharacter{3B2}{\beta}            % β
\DeclareUnicodeCharacter{3B3}{\gamma}           % γ
\DeclareUnicodeCharacter{3B4}{\delta}           % δ
\DeclareUnicodeCharacter{3B5}{\varepsilon}      % ε
\DeclareUnicodeCharacter{3C6}{\varphi}          % φ
\DeclareUnicodeCharacter{3B7}{\eta}             % η
\DeclareUnicodeCharacter{3B8}{\theta}           % θ
\DeclareUnicodeCharacter{3BC}{\mu}              % μ
\DeclareUnicodeCharacter{3BD}{\nu}              % ν
\DeclareUnicodeCharacter{3BE}{\xi}              % ξ
\DeclareUnicodeCharacter{3C0}{\pi}              % π 
\DeclareUnicodeCharacter{3C1}{\rho}             % ρ
\DeclareUnicodeCharacter{3C3}{\sigma}           % σ
\DeclareUnicodeCharacter{2026}{\ldots}          % …
\DeclareUnicodeCharacter{2112}{\mathcal{L}}     % ℒ
\DeclareUnicodeCharacter{2113}{\ell}            % ℓ
\DeclareUnicodeCharacter{2115}{\mathbb{N}}      % ℕ
\DeclareUnicodeCharacter{211D}{\mathbb{R}}      % ℝ
\DeclareUnicodeCharacter{2133}{\mathcal{M}}     % ℳ
\DeclareUnicodeCharacter{2192}{\rightarrow}     % →
\DeclareUnicodeCharacter{21A6}{\mapsto}         % ↦
\DeclareUnicodeCharacter{2203}{\exists}         % ∃ 
\DeclareUnicodeCharacter{2208}{\in}             % ∈
\DeclareUnicodeCharacter{220F}{\prod}           % ∏
\DeclareUnicodeCharacter{2211}{\sum}            % ∑
\DeclareUnicodeCharacter{2218}{\circ}           % ∘
\DeclareUnicodeCharacter{221E}{\infty}          % ∞
\DeclareUnicodeCharacter{222A}{\cup}            % ∪
\DeclareUnicodeCharacter{2254}{\vcentcolon=}    % ≔ 
\DeclareUnicodeCharacter{225C}{\triangleq}      % \coloneqq
\DeclareUnicodeCharacter{2260}{\neq}            % ≠ 
\DeclareUnicodeCharacter{2261}{\equiv}          % ≡
\DeclareUnicodeCharacter{2282}{\subset}         % ⊂
\DeclareUnicodeCharacter{2286}{\subseteq}       % ⊆ 
\DeclareUnicodeCharacter{228A}{\subsetneq}      % ⊊
\DeclareUnicodeCharacter{22C5}{\cdot}           % ⋅
\DeclareUnicodeCharacter{22EF}{\cdots}          % ⋯
\DeclareUnicodeCharacter{2308}{\lceil}          % ⌈ 
\DeclareUnicodeCharacter{2309}{\rceil}          % ⌉
\DeclareUnicodeCharacter{230A}{\lfloor}         % ⌊
\DeclareUnicodeCharacter{230B}{\rfloor}         % ⌋
\DeclareUnicodeCharacter{27E8}{\langle}         % ⟨
\DeclareUnicodeCharacter{27E9}{\rangle}         % ⟩
\DeclareUnicodeCharacter{2A7D}{\leqslant}       % ⩽
\DeclareUnicodeCharacter{2A7E}{\geqslant}       % ⩾

\definecolor{veroeditcolor}{RGB}{255,210,210}
\definecolor{oliviereditcolor}{RGB}{150,200,55}
\definecolor{santieditcolor}{RGB}{210,210,255}

\newcommand{\hc}[4]{H(X^{{#1},{#2}}_{#3}\mid X^{{#1},{#2}}_{{#4}},\dots,X^{{#1},{#2}}_{{#3}-1})}
\newcommand{\h}[3]{H(X^{#1,#2}_{0},\dots,X^{#1,#2}_{#3})}
\newcommand{\probi}[5]{p(X^{#1,#2}_{#3}=a_{0},\dots,X^{#1,#2}_{#4}=a_{#5})}

\begin{document}
\title
{Rauzy complexity and block entropy}
\author{Verónica Becher \and Olivier Carton \and Santiago Figueira}
\date{\today \ - Revised version}
\maketitle

\begin{abstract}
In 1976, Rauzy studied two complexity functions, $\ubetar$ and $\obetar$, for infinite sequences over a finite alphabet. The function $\ubetar$ achieves its maximum precisely for Borel normal sequences, while $\obetar$ reaches its minimum for sequences that, when added to any Borel normal sequence, result in another Borel normal sequence.
We establish a connection between Rauzy's complexity functions, $\ubetar$ and $\obetar$,  and the notions of non-aligned block entropy, $\unh$ and $\onh$, by providing sharp upper and lower bounds for $\unh$ in terms of $\ubetar$, and sharp upper and lower bounds for $\onh$ in terms of $\obetar$. We adopt a probabilistic approach by considering an infinite sequence of random variables over a finite alphabet. The proof relies on a new characterization of non-aligned block entropies, $\onh$ and $\unh$, in terms of Shannon's conditional entropy. The bounds imply that sequences with $\onh = 0$ coincide with those for which $\obetar = 0$.
We also  show that the non-aligned block entropies, $\unh$ and $\onh$, are essentially subadditive.
\end{abstract}

% \tableofcontents

\section{Introduction and Statement of Results}

In 1976 Rauzy~\cite{Rauzy76} introduced the functions  $\betar$ and $\gammar$   from the set of infinite sequences to real numbers in the interval $[0, 1/2]$ to classify infinite sequences.
He proved two main results  that can be expressed both with the $\betar$ or the $\gammar$ functions. One is that the sequences  with maximum value of $\betar$  are exactly the binary expansions of Borel normal numbers.  The other is that the sequences with   minimum value  of $\betar$, which is $0$, are exactly the numbers that preserve normality when added to a Borel normal number. Rauzy calls these sequences deterministic.

The first main result of this note, Theorem~\ref{thm:main},  characterizes the Rauzy functions in terms of the already known notion of non-aligned block entropy studied in \cite{KozachinskiyShen2021}. 
Our second main result  proves, with full generality, that the 
 properties that  Rauzy functions have when they take extreme  values (maximum and minimum) also hold for the non-aligned block entropy: Theorem~\ref{thm:subadditive} proves that the non-aligned block entropy, as a  function from real numbers to the unit interval, is subadditive.  Finally, in Theorem~\ref{thm:hc_equals_h} we prove that the  non-aligned block entropy, as a function from  infinite sequences to the unit interval, can be expressed  as a limit of  Shannon's conditional  entropy. 
  
% \santi{I added definitions of length and positions of a word, that were used but no defined.}
Let $\B=\{0,1\}$. For $w\in\B^*$, the length of $w$ is denoted by $|w|$, and
we use $w[i]$ to denote the bit $w$ of in position $i$, where the positions of
$w$ are numbered $0,\dots,|w|-1$. For a finite or infinite sequence $w$ of
symbols in $\B$ we use $w[i:j]$ to denote the subsequence of $w$ from
position $i$ to position $j$, and we use left and/or right round bracket to
exclude  the extremes. For each positive integer $ℓ ⩾ 1$ and for each
word~$w∈\B^*$, the Rauzy functions  $\betarl(w)$ and $\gammarl(w)$ are
defined by
% \olivier{In the following definitions, the summations are used to count
%   the number of mismatches, as explained in the paragraph below.  I would
%   prefer a more explicit definitions like
%   \begin{align*}
%     \betar_ℓ(w) & = \frac{1}{|w|}\min_{f : A^ℓ → A}\#\{i : w[i] ≠ f(w[i+1:i+ℓ])\} \\
%     \gammar_ℓ(w) & = \frac{1}{|w|}\min_{f : A^ℓ → A}\#\{i : w[i] ≠ f(w[i-ℓ:i-1])\}
%   \end{align*}
% }
% \santi{OK! I think we should put the range of the $i$, namely
%   \begin{align*}
%     \betar_ℓ(w) & = \frac{1}{|w|}\min_{f : A^ℓ → A}\#\{i \in\{0,\dots,|w|-\ell-1\}: w[i] ≠ f(w[i+1:i+ℓ])\} \\
%     \gammar_ℓ(w) & = \frac{1}{|w|}\min_{f : A^ℓ → A}\#\{i \in\{\ell,\dots,|w|-1\} : w[i] ≠ f(w[i-ℓ:i-1])\}
%   \end{align*}
%   }
% \begin{align*}
%   \betarl(w) & \coloneqq \min_{f:\B^ℓ → \B} \frac{1}{|w|}
%                 ∑_{i=0}^{|w|-ℓ-1} (1-δ_{w[i],f(w[i+1:i+ℓ])}) \\
%   \gammarl(w) & \coloneqq \min_{f:\B^ℓ → \B} \frac{1}{|w|}
%                 ∑_{i=0}^{|w|-ℓ-1}(1-δ_{f(w[i:i+ℓ)),w[i+ℓ]}), 
% \end{align*}
% \sidesanti{I put Olivier's definition but I added the range of $i$. Please
%   check}
% \sideolivier{I checked range of~$i$}
\begin{align*}
    \betar_ℓ(w) & \coloneqq \frac{1}{|w|-ℓ}\min_{f : A^ℓ → A}\#\bigl\{i : w[i] ≠ f(w(i:i+ℓ])\big\} \\
    \gammar_ℓ(w) & \coloneqq \frac{1}{|w|-ℓ}\min_{f : A^ℓ → A}\#\big\{i  : w[i] ≠ f(w[i-ℓ:i)\big\},
\end{align*}
where in the first definition the range of $i$ is $\{0,\dots,|w|-\ell-1\}$ and in the second one is $\{\ell,\dots,|w|-1\}$.
Therefore, the summation gives the number of mismatches
between the symbol at position~$i$ (respectively at position $i+ℓ$) and
the symbol given by the function~$f$ applied to the factor of length~$ℓ$
after (respectively before) position~$i$. 

Rauzy extends these definitions\footnote{In~\cite{Rauzy76} Rauzy uses $\beta'_\ell$ and for $\gammarl$. He gives his theorem on  $\obetar$ and $\ubetar$ and just points \cite[Remarque page 212]{Rauzy76} 
that    $\obetar$ and $\ogammar$  are $0$ on the same  sequences, and  $\ubetar$ and $\ugammar$ are  $1/2$ on the same sequences.
}
to $\ubetarl,\obetarl,\ubetar,\obetar:\B^ℕ→[0,1/2]$ and
$\ugammarl,\ogammarl,\ugammar,\ogammar:\B^ℕ→[0,1/2]$,
\begin{align*}
  \ubetarl(x) \coloneqq \liminf_{n → ∞}\betarl(x[0:n))
  & \quad\text{and}\quad
  \obetarl(x) \coloneqq \limsup_{n → ,∞}\betarl(x[0:n)), 
  \\
  \ubetar(x) \coloneqq \lim_{ℓ → ∞}\ubetarl(x)
  & \quad\text{and}\quad
  \obetar(x) \coloneqq \lim_{ℓ → ∞}\obetarl(x), 
  \\
  \ugammarl(x) \coloneqq \liminf_{n → ∞}\gammarl(x[0:n))
  & \quad\text{and}\quad
  \ogammarl(x) \coloneqq \limsup_{n → ∞}\gammarl(x[0:n)), 
  \\
  \ugammar(x) \coloneqq \lim_{ℓ → ∞}\ugammarl(x)
  & \quad\text{and}\quad
  \ogammar(x) \coloneqq \lim_{ℓ → ∞}\ogammarl(x).
\end{align*}
These limits exist for $\ell\to\infty$ because $\ubetarl(x)$, $\obetarl(x)$,
$\ugammarl(x)$ and $\ogammarl(x)$ are decreasing as functions in~$ℓ$
while~$x$ is fixed. 
Rauzy 
 \cite[page 212]{Rauzy76} says
 that in a more general context the two functions may not coincide.\footnote{
Based on the correspondence between stationary stochastic processes 
and measure-preserving transformations in ergodic theory, 
Rauzy  \cite[page 212]{Rauzy76} cites Krengel's work \cite{Krengel1970}, 
 to provide an example where the functions  $\beta$ and $\gamma$
 do not coincide. Specifically, in \cite[Corollary page 137]{Krengel1970},
Krengel establishes the existence of a two-state process that 
 is forward deterministic and backward completely nondeterministic, 
and for which the associated shift is nonsingular, ergodic 
and has an infinite invariant measure.}

We prove that the functions do not coincide in the simple space $\B^ℕ$, the space of infinite sequences with  equiprobable symbols.
% \sidesanti{$\B^n$?}
% \sideolivier{yes it should be $\B^ℕ$}

% In \cite[page 212]{Rauzy76} Rauzy remarks\footnote{
% A probability space satisfying minimal regularity conditions
% the ergodic theory of measure preserving  transformations is, in a certain sense, 
% equivalent to the theory of stationary processes.
% Rauzy \cite[page 212]{Rauzy76} cites the article by Krengel 
%\cite{Krengel1970}, possibly in relation to the following  \cite[Corollary page 137]{Krengel1970}:
%  there exists a process with two states which is forward deterministic and backward completely nondeterministic, 
%and for which the associated shift is nonsingular and ergodic and has an infinite invariant measure.}
%  that  in  general probability spaces the functions $\obetar$ and~$\ogammar$, not necessarily  coincide, nor do $\ubetar$ and $\ugammar$.
% We show that 
% %it is worth distinguishing between 
% the functions are indeed different in $\B^{\N}$, which as probability space it has  uniform probability.
% \santi{what do you mean by "with the unif prob"?}
% \vero{I changed the wording, does it read better?
% Feel free to change it.}

\begin{proposition}\label{prop:gamma-beta}
  The  functions $\ubetar$ and $\ugammar$  are different, and so are
  $\obetar$ and $\ogammar$.
\end{proposition}

We give the proof of this proposition in Section \ref{sec:gamma-beta}.

For words $w$ and $u$ in $\B^*$, the number $\occ{w}{u}$ of
occurrences of~$u$ in~$w$ is 
% and the number $\alocc{w}{u}$ of
% aligned occurrences of~$u$ in~$w$ are respectively defined by
\begin{align*}
  \occ{w}{u}   & \coloneqq \#\{ i : w[i:i+|u|) = u \}.
  % \alocc{w}{u} & \coloneqq \#\{ i : w[i:i+|u|) = u \text{ and } i = 0 \bmod |u|\}.
\end{align*}
For instance $\occ{0000}{00} = 3$.
%while $\alocc{0000}{00} = 2$.  The
% notation $\alocc{w}{u}$ is usually used when the length~$|w|$ is a multiple
% of the length~$|u|$ but this is not required by the definition.
% 
% Now we consider the notion of aligned  entropy  of a given word, also known as "aligned block entropy"  which is defined by the frequency of non-overlapping blocks in the given word.  
% 
% For each positive integer $ℓ ⩾ 1$, the normalized aligned
% $ℓ$-entropy~$\ahl(w)$ of a word~$w∈\B^*$ is defined as follows:
% \begin{displaymath}
%   \ahl(w)  \coloneqq -\frac{1}{ℓ} ∑_{u ∈ \B^ℓ} f_u \log f_u 
%   \quad\text{where}\quad
%   f_u = \frac{\alocc{w}{u}}{⌊|w|/ℓ⌋} 
% \end{displaymath}
% with the usual convention $0\log 0 = 0$.  By \emph{normalized}, we mean
% that the summation is divided by~$ℓ$ so that the inequalities $0 ⩽ \ahl(w)
% ⩽ 1$ hold.  
% 
% The aligned entropies $\uah(x)$ and~$\oah(x)$ of a
% sequence~$x$ are respectively defined as follows.
% \begin{alignat*}{2}
%   \uah(x) & \coloneqq \liminf_{ℓ → ∞} \uahl(x) & \quad\text{where}\quad
%   \uahl(x) & \coloneqq \liminf_{n → ∞} \ahl(x[0:n)) \\
%   \oah(x) & \coloneqq \liminf_{ℓ → ∞} \oahl(x) & \quad\text{where}\quad
%   \oahl(x) & \coloneqq \limsup_{n → ∞} \ahl(x[0:n)) 
% \end{alignat*}
%
The non-aligned $\ell$-block entropy of a word $w∈\B^*$ is defined by 

\begin{displaymath}
  \nhl(w)  \coloneqq -\frac{1}{ℓ} ∑_{u ∈ \B^ℓ} f_u \log f_u 
  \quad\text{where}\quad
  f_u = \frac{\occ{w}{u}}{|w|-|u|+1}. 
\end{displaymath}
We assume $\log$ is the
logarithm in base~$2$ and it is assumed the usual convention that $0\log 0
= 0$.

The non-aligned block entropies $\unh(x)$ and~$\onh(x)$ of a
sequence~$x\in\B^ℕ$ are respectively defined as follows,
\begin{alignat*}{2}
  \unh(x) & \coloneqq \lim_{ℓ → ∞} \unhl(x) & \quad\text{where}\quad
  \unhl(x) & \coloneqq \liminf_{n → ∞} \nhl(x[0:n)) \\
  \onh(x) & \coloneqq \lim_{ℓ → ∞} \onhl(x) & \quad\text{where}\quad
  \onhl(x) & \coloneqq \limsup_{n → ∞} \nhl(x[0:n)). 
\end{alignat*}

Kozachinsky and Shen  prove \cite[Theorem 8]{KozachinskiyShen2021}
 that $\lim_{ℓ → ∞}\unh_\ell$ always exists.
 Ziv and Lempel in \cite{ZivLempel78} prove  that $\lim_\ell\onh_{ℓ → ∞}$ always exists.
% A proof that for every $x ∈ \B^ℕ$, $\orhod(x) = \oah(x)$ appears in
%~\cite{Sheinwald94}.  The inequality $\oah(x) ⩽ \orhod(x)$.  first appeared
% in~\cite{ZivLempel78}.
% Moving from $\orhod(x) = \oah(x)$ to $\urhod(x) = \uah(x)$ is not difficult.
It is well known that $\unh$ and $\onh$ have several alternative but equivalent formulations, based on finite-state compressors, finite-state measures and finite-state dimension, see  \cite{KozachinskiyShen2021} and the references cited there.

The following is our first main result.
It uses the Shannon entropy function 
$\entropyprob:[0,1]→[0,1]$ for a Bernoulli random variable with probability 
$\alpha$ of being $1$, and 
$1-\alpha$  of being $0$. Let
$\entropyprob(\alpha) \coloneqq -\alpha \log \alpha - (1-\alpha)\log (1-\alpha)$. 
 The graph of the function~$\entropyprob$ is pictured in
Figure~\ref{fig:function-h}.
% We assume $\log$ is the
% logarithm in base~$2$ and it is assumed the usual convention that $0\log 0
% = 0$. 

\begin{figure}[tbp]
  \begin{center}
    \includegraphics[scale=0.8]{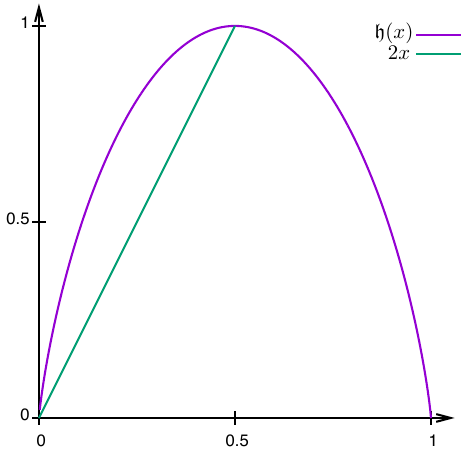}
  \end{center}
  \caption{Graph of the function~$\entropyprob$}
  \label{fig:function-h}
\end{figure}

\begin{theorem}
\label{thm:main}
\begin{enumerate}

\item For every $x∈\B^ℕ$,
  \begin{alignat*}{2}
   2\ugammar(x) & ⩽ \unh(x)  ⩽ \entropyprob(\ugammar(x)) 
    &\text{\qquad and\qquad }
   2\ogammar(x) ⩽ &\onh(x)  ⩽ \entropyprob(\ogammar(x)), \\
   2\ubetar(x) & ⩽ \unh(x)  ⩽ \entropyprob(\ubetar(x)) 
    &\text{\qquad and\qquad }
   2\obetar(x) ⩽ &\onh(x)  ⩽ \entropyprob(\obetar(x)). 
  \end{alignat*}

\item These bounds are sharp.
\end{enumerate}
\end{theorem}
Item 2 of Theorem \ref{thm:main} is made 
%more 
 explicit in Proposition \ref{prop:sharp}.

The two extreme cases $\ogammar(x) = 0$ and $\ugammar(x) = 1/2$ deserve
some comments.  
% \note{S: clarify that deterministic is a definition}
% Rauzy uses the term {\em deterministic} for those sequences having  $\ogammar$ equal to  $0$.
For example, any Sturmian sequence $x$ satisfies that $\ogammar(x) = 0$.
Since $\entropyprob(0) = 0$, a sequence~$x ∈ \B^ℕ$
satisfies $\ogammar(x) = 0$ if and only if it satisfies $\onh(x) =
0$.  Since $\entropyprob(1/2) = 1$, a sequence~$x ∈ \B^ℕ$ satisfies
$\ugammar(x) = 1/2$ if and only if it satisfies $\unh(x) = 1$, that is,
$x$ is a Borel normal sequence. This characterization of Borel normality is already
proved in~\cite{Rauzy76}; for some other characterizations of Borel normality see \cite{BecherCarton18} and the references there.

Our second main result is as follows.
For $x,y ∈ \B^ℕ$ we write $x+y$ for the binary expansion of the
addition of the real numbers denoted by $x$ and $y$.

\begin{theorem} \label{thm:subadditive}
For every $x,y∈\B^ℕ$,
  \begin{align*}
    \unh(x)-\onh(y) & ⩽ \unh(x+y) ⩽ \unh(x)+\onh(y), \\
    \onh(x)-\onh(y) & ⩽ \onh(x+y) ⩽ \onh(x)+\onh(y).
  \end{align*}
\end{theorem}

By Theorem~\ref{thm:main}, $\ugammar(x) = 1/2$ is equivalent to
$\unh(x) = 1$ which is, in turn, equivalent to $x$ being Borel normal.  In the
same way, $\ogammar(y) = 0$ is equivalent to $\onh(y) = 0$.  In the
special case $\unh(x) = 1$ and $\onh(y) = 0$, Theorem~\ref{thm:subadditive}
states that $\unh(x+y) = \unh(x) = 1$ which exactly means that adding~$y$
such that $\onh(y) = 0$ to a Borel normal number~$x$ preserves  normality.

To prove Theorems~\ref{thm:main} and~\ref{thm:subadditive}, we adopt the probabilistic approach of considering random variables defined over $x∈\B^ℕ$ and a characterization of
the non-aligned block entropy \cite{KozachinskiyShen2021} in terms of the Shannon's  entropy. In particular, Theorem~\ref{thm:main} relies in a new characterization of the 
non-aligned block entropy $\onh(x)$ and $\unh(x)$ in terms of Shannon's conditional  entropy.

Recall that in general for discrete random variables $X$ and $Y$ with support sets
 ${\mathcal X}$ and ${\mathcal Y}$ respectively, the
Shannon's entropy of $X$ is defined by
\[
H(X)\coloneqq-\sum _{x\in {\mathcal {X}}}p(X=x)\log p(X=x)
\]
and Shannon's conditional  entropy of $X$ given $Y$ is defined by

\begin{align*}
H(X\mid Y)&\coloneqq -\sum_{y\in\mathcal{Y}}p(Y=y)\sum_{x\in\mathcal{X}} p(X=x\mid Y =y)\log p(X=x\mid Y=y) 
\\
&=-\sum_{x\in\mathcal{X},y\in\mathcal{Y}}p(X=x,Y=y)\log \frac{p(X=x,Y=y)}{p(Y=y)}.
\end{align*}
% where $\mathcal{X}$ and $\mathcal{Y}$ are the support sets of $X$ and $Y$ respectively.

For our purposes we fix a sequence $x\in\B^ℕ$ and consider $w=x[0:n)$ and a positive integer $\ell<n$. 
We consider $\B$-valued possibly dependent random variables $X^{n,\ell}_0,\dots,X^{n,\ell}_{\ell}$ 
% with probabilities 
whose joint probability is defined by
\begin{equation}\label{eqn:distribution}
\probi{n}{\ell}{0}{\ell}{\ell}\coloneqq
\frac{\occ{w}{a_0\cdots a_{\ell}}}{n-\ell},
\end{equation}
and by marginalizing we obtain that
\begin{equation*}
p(X_i^{n,\ell}=a)=\frac{\occ{w[i:n-\ell+i]}{a}}{n-\ell}.
\end{equation*}
In other words, these random variables form a randomly
selected non-aligned block of length $\ell+1$ inside $w$ so that for $i\in\{0,\dots,\ell\}$ we have $X_i^{n,\ell}=a$ if the $(i+1)$-th symbol of such block is $a$. Alternatively, we could have just defined one single random variable with values in $\B^{\ell+1}$. We have chosen to use multiple, but dependent, single bit random variables for reasons that will become clear later.

Then the Shannon's entropy becomes 
$$
H(X^{n,\ell}_0,\ldots ,X^{n,\ell}_\ell)=\sum \probi{n}{\ell}{0}{\ell}{\ell}\log \probi{n}{\ell}{0}{\ell}{\ell},
$$
where the range of the sum is the set of $a_0,\dots,a_\ell\in\B$.
Observe that for $w=x[0:n)$ we have
$h_{\ell+1}(w)= \h{n}{\ell}{\ell}/(\ell+1)$ and, consequently, for each $\ell$,
\begin{align}
\unh_\ell(x)&=\liminf_{n → ∞} \h{n}{\ell}{\ell}/(\ell+1),\label{unh:H}
\\
\onh_\ell(x)&=\limsup_{n → ∞} \h{n}{\ell}{\ell}/(\ell+1).\label{onh:H}
\end{align}
We define new notions of {\em conditional block entropy} $\uhc:\B^{\mathbb N}\to [0,1]$ and $\ohc:\B^{\mathbb N}\to [0,1]$ defined by
\begin{align*}
\uhc(x)&\coloneqq\lim_{ℓ → ∞}\liminf_{n\to\infty} \hc{n}{\ell}{\ell}{0}\\
\ohc(x)&\coloneqq\lim_{ℓ → ∞}\limsup_{n\to\infty} \hc{n}{\ell}{\ell}{0}.
\end{align*}
%
% \santi{
% The fact that $\lim_\ell\onh_\ell$ exists is in J. Ziv and A. Lempel. Compression of individual sequences via variable-rate
% coding.  
% The fact that $\lim_\ell\unh_\ell$ exists is in \cite[Thm 8]{KozachinskiyShen2021}.
% So we have 
% \begin{align*}
%   \unh(x)  \coloneqq \lim_{ℓ → ∞} \unhl(x)  \\
%   \onh(x) \coloneqq \lim_{ℓ → ∞} \onhl(x)  
% \end{align*}
% Then, the fact that the limits of $\uhc$ and $\ohc$ exist follows from the argument inside the proof of Thm \ref{thm:hc_equals_h}. 
%
% How do we want to present the definition of $\uhc$ and $\ohc$?
% }
%
The  existence of the limits defining 
$\uhc$ and 
$\ohc$ is shown inside the proof of the next Theorem~\ref{thm:hc_equals_h}, which is our third and final main result in this note.

\begin{theorem}\label{thm:hc_equals_h}
$\unh=\uhc$ and $\onh=\ohc$.
\end{theorem}

The rest of the paper contains, one per section, the proofs of Theorem~\ref{thm:hc_equals_h}, Theorem~\ref{thm:subadditive},
Theorem~\ref{thm:main} and 
 Proposition~\ref{prop:gamma-beta}.
 
\section{Proof of Theorem~\ref{thm:hc_equals_h}}\label{sec:hc_equals_h}

% Let $X$ and $Y$ be random variables. In general, the conditional entropy of $X$ given $Y$ is defined by:
% \begin{align*}
% H(X\mid Y)&= -\sum_{y\in\mathcal{Y}}p(Y=y)\sum_{x\in\mathcal{X}} p(X=x\mid Y =y)\log p(X=x\mid Y=y) 
% \\
% &=-\sum_{x\in\mathcal{X},y\in\mathcal{Y}}p(X=x,Y=y)\log \frac{p(X=x,Y=y)}{p(Y=y)}
% \end{align*}
% where $\mathcal{X}$ and $\mathcal{Y}$ are the support sets of $X$ and $Y$ respectively.

We have to show that
$$
\lim_{ℓ → ∞}\liminf_{n → ∞} \h{n}{\ell}{\ell}/(\ell+1)=\lim_{ℓ → ∞}\liminf_{n\to\infty} \hc{n}{\ell}{\ell}{0}.
$$
and the analogue with $\limsup$. The main tool will be the chain rule for multiple random variables, which states that 
$$
\h{n}{\ell}{\ell} = \sum_{i=0}^\ell \hc{n}{\ell}{i}{0}.
$$
For the inequality $\geqslant$ it suffices to show 
$$
\liminf_{n → ∞} \h{n}{\ell}{\ell}/(\ell+1)\geqslant\liminf_{n\to\infty} \hc{n}{\ell}{\ell}{0}.
$$
for any $\ell$, which will basically follow from the chain rule and the fact that for large enough $n$ all terms of the sum (there are $\ell+1$ many) are close to $\hc{n}{\ell}{\ell}{0}$. 

For the inequality $\leqslant$ the argument is more complex. 
We have to increase $\ell$ to $\ell'$ in the left hand side of the inequality
 to keep it true, but this is allowed since the statement has to hold when $\ell$ goes to infinity. We also use the chain rule for proving this inequality split sum resulting from the chain rule to bound each part separately.
In what follows we give the formal proof and make explicit all the needed properties.

The following proposition follows from the definition of the random variables.

\begin{proposition}\label{prop:probabilities}
For $0⩽j⩽i⩽\ell$ we have
\begin{align*}
\probi{n}{\ell}{i-j}{i}{j}=&
\frac{\occ{w[i-j:n-\ell+i)}{a_{0}\cdots a_{j}}}{n-\ell}.
\end{align*}
\end{proposition}

The following corollaries will be used in the proof of Theorem~\ref{thm:hc_equals_h} and they are straightforward.

\begin{corollary}\label{cor:probi_probk}
Let $j\in [0, \ell]$ and let $i,k\in [j,\ell]$. Then,
$$
\lim_{n\to\infty}\left|\probi{n}{\ell}{i-j}{i}{j}-\probi{n}{\ell}{k-j}{k}{j}\right|=0.
$$
\end{corollary}
\begin{proof}
By Proposition~\ref{prop:probabilities} we have 
\[
\left|\probi{n}{\ell}{i-j}{i}{j}-\probi{n}{\ell}{k-j}{k}{j}\right|⩽\frac{2\ell}{n-\ell}.
\]
and this concludes the proof.\end{proof}

\begin{corollary}\label{cor:conditionalEntropy}
Let $j\in (0, \ell]$, $i,k\in [j,\ell]$. Then
$$
\lim_{n\to\infty}\left|\hc{n}{\ell}{k}{k-j}-\hc{n}{\ell}{i}{i-j}\right|=0.
$$
\end{corollary}
\begin{proof}
It follows from Corollary~\ref{cor:probi_probk} and the continuity of the entropy function.
\end{proof}

\begin{corollary}\label{cor:ell_prime}
If $\ell'\geqslant\ell$ then
$$
\lim_{n\to\infty}\left|
\probi{n}{\ell'}{\ell'-\ell}{\ell'}{\ell}
- 
\probi{n}{\ell}{0}{\ell}{\ell}
\right|=0.
$$
\end{corollary}

\begin{proof}
By Proposition~\ref{prop:probabilities},
$$
\probi{n}{\ell'}{\ell'-\ell}{\ell'}{\ell}=
\frac{\occ{w[\ell'-\ell:n)}{a_{0}\cdots a_{\ell}}}{n-\ell'}
$$
and
$$
\probi{n}{\ell}{0}{\ell}{\ell}=
\frac{\occ{w}{a_{0}\cdots a_{\ell}}}{n-\ell}.
$$
and this concludes the proof.
\end{proof}

\begin{corollary}\label{cor:H_ell_ell_prime}
If $\ell'\geqslant\ell$ then
$$
\lim_{n\to\infty}\left|\hc{n}{\ell'}{\ell'}{\ell'-\ell}-\hc{n}{\ell}{\ell}{0}\right|=0.
$$
\end{corollary}
\begin{proof}
It follows from Corollary~\ref{cor:ell_prime} and the continuity of the entropy function.
\end{proof}

\begin{proposition}\label{from_ell_prime_to_ell}
Let $\ell⩽i⩽\ell'$. For every $\epsilon>0$ and large enough $n$ we have
$$
\hc{n}{\ell'}{\ell'}{\ell'-i} 
< 
\hc{n}{\ell}{\ell}{0}+\epsilon.
$$
% $$
% \lim_{n\to\infty}\hc{n}{\ell'}{\ell'}{\ell'-i} ⩽\lim_{n\to\infty}\hc{n}{\ell}{\ell}{0}.
% $$
\end{proposition}

\begin{proof}
% The first inequality follows by dropping conditions (and in fact it holds for every $n$). The second one follows from Corollary~\ref{cor:H_ell_ell_prime}.
On the one hand, by dropping conditions we have that
\begin{align*}
\hc{n}{\ell'}{\ell'}{\ell'-i}&⩽\hc{n}{\ell'}{\ell'}{\ell'-\ell}.
\end{align*}
On the other, by Corollary~\ref{cor:H_ell_ell_prime}, for any $\epsilon$ and large enough $n$ it holds that
$$
\hc{n}{\ell'}{\ell'}{\ell'-\ell}<\hc{n}{\ell}{\ell}{0}+\epsilon.
$$
and this concludes the proof.
\end{proof}

\begin{proof}[Proof of Theorem~\ref{thm:hc_equals_h}]
We show it for limsup. It suffices to show that for any $\epsilon>0$ and large enough~$n$, 

\begin{equation}\label{eqn:lower}
\h{n}{\ell}{\ell}/(\ell+1) ⩾\hc{n}{\ell}{\ell}{0}-\epsilon,
\end{equation}
% and for large enough $\ell'\geqslant\ell$,
for $\ell<n$ and that for large enough $\ell'\in[\ell,n)$,
\begin{equation}\label{eqn:upper}
\h{n}{\ell'}{\ell'}/(\ell'+1)  ⩽\hc{n}{\ell}{\ell}{0} + \epsilon.
\end{equation}
%Let us see~\eqref{eqn:lower}. 
By the chain rule we know that 
$$
\h{n}{\ell}{\ell} = \sum_{i=0}^\ell \hc{n}{\ell}{i}{0}.
$$
Let $\epsilon>0$. By Corollary~\ref{cor:conditionalEntropy} if $n$ is large enough and $i⩽\ell$ then
$$
\hc{n}{\ell}{i}{0}⩾\hc{n}{\ell}{\ell}{\ell-i}-\epsilon
$$
and, by adding more conditions to the entropy, 
\begin{equation}\label{eqn:adding_conditions}
\hc{n}{\ell}{\ell}{\ell-i}⩾\hc{n}{\ell}{\ell}{0}.
\end{equation}
Therefore, we obtain \eqref{eqn:lower},  that is,
\begin{align*}
\h{n}{\ell}{\ell} &
\geqslant
(\ell+1) \left(\hc{n}{\ell}{\ell}{0}-\epsilon\right).
\end{align*}

%Let us see~\eqref{eqn:upper}. 
Now observe that the inequality in~\eqref{eqn:adding_conditions} is not reversed in general and so we cannot repeat the strategy above. 
By Proposition~\ref{from_ell_prime_to_ell}, for any $\ell'\geqslant\ell$ and any $i\in[\ell,\ell']$ we have that for $n$ large enough,
$$
\hc{n}{\ell'}{\ell'}{\ell'-i} ⩽\hc{n}{\ell}{\ell}{0}+\epsilon/3.
$$
By Corollary~\ref{cor:conditionalEntropy} for large enough $n$ and all $i⩽\ell'$ we have
$$
\hc{n}{\ell'}{i}{0} ⩽\hc{n}{\ell'}{\ell'}{\ell'-i}+\epsilon/3
$$
and hence,
$$
\hc{n}{\ell'}{i}{0} ⩽\hc{n}{\ell}{\ell}{0} + 2\epsilon/3.
$$
Again by the chain rule we have
\begin{align*}
\h{n}{\ell'}{\ell'} &= \sum_{i=0}^{\ell'} \hc{n}{\ell'}{i}{0} 
\\
&= \sum_{i=0}^{\ell-1} \hc{n}{\ell'}{i}{0} +  \sum_{i=\ell}^{\ell'} \hc{n}{\ell'}{i}0
\\
& ⩽\ell + (\ell'+1) \left(\hc{n}{\ell}{\ell}{0} + 2\epsilon/3\right).
\end{align*}
If $\ell'$ is such that $\ell/(\ell'+1)<\epsilon/3$ then we conclude  \eqref{eqn:upper}, that is
$$
\h{n}{\ell'}{\ell'}/(\ell'+1)  ⩽\hc{n}{\ell}{\ell}{0} + \epsilon.
$$ 
Then, from \eqref{eqn:lower} and \eqref{eqn:upper}
we conclude 
\[
\lim_{\ell\to\infty}
\limsup_{n\to \infty} \h{n}{\ell}{\ell}/(\ell+1) =
\lim_{\ell\to\infty}
\limsup_{n\to\infty }\hc{n}{\ell}{\ell}{0}.
\]
Therefore, using 
$h_{\ell+1}(x[0:n))= \h{n}{\ell}{\ell}/(\ell+1)$, we obtain
$\onh(x)=\ohc(x).$
\medskip

The case for $\liminf$ is similar and allows us to conclude $\unh(x)=\uhc(x).$
\end{proof}

\section{Proof of Theorem ~\ref{thm:subadditive}}\label{sec:subadditive}

\begin{lemma} \label{lem:simplification}
\
\begin{enumerate}
\item If
  $\unh(x+y) ⩽ \unh(x)+\onh(y)$ for every $x,y ∈ \B^ℕ$,
  then $\unh(x)-\onh(y) ⩽ \unh(x+y)$ for every $x,y ∈ \B^ℕ$.
\item If 
  $\onh(x+y) ⩽ \onh(x)+\onh(y)$ for every $x,y ∈ \B^ℕ$, then 
   $\onh(x)-\onh(y) ⩽ \onh(x+y)$ for every $x,y ∈ \B^ℕ$.
\end{enumerate}
\end{lemma}
\begin{proof}
  We show item 1; item 2 is analogous. We apply the hypothesis with $x' =
  x+y$ and $y' = -y$ to get $\unh(x) ⩽ \unh(x+y) + \onh(-y) = \unh(x+y)
  + \onh(y)$.  The last equality comes from $\onh(-y) = \onh(y)$.
\end{proof}
\medskip

\begin{proof}[Proof of Theorem \ref{thm:subadditive}]
We must relate the non-aligned block entropy of $x$, $y$, and $x + y$.
  By Lemma~\ref{lem:simplification} it is sufficient to prove the two
  $\unh(x+y) ⩽ \unh(x)+\onh(y)$ and $\onh(x+y) ⩽ \onh(x)+\onh(y)$.
 We start with the first of the two inequalities.

We fix sequences $x,y,z\in\B^ℕ$ such that $z=x+y$. Let $w=z[0:n)$, $u=x[0:n)$, $v=y[0:n)$, and let $\ell$ and $n$ such that  $\ell<n$.  If we split the binary representations for $w$, $u$ and $v$ into blocks of size $ \ell$, then the block for $w$ is determined, up 
%to finite number of  bits related to
the carry bit status, by the corresponding blocks for $u$ and $v$.

% \sideolivier{The fact that the random variable~$C$ is not formely defined
%   made it difficult for me to read this part.}
%   \sidesanti{I don't know what is missing for the formal definition. Again these are dependent variables, and $C$ is just a 0 or a 1. The important thing is the distribution, which is what I tried to define as $p$ below here}
We consider the $\B$-valued possibly dependent random variables $C$, $X^{n,\ell}_0,\dots,X^{n,\ell}_{\ell}$,
$Y^{n,\ell}_0,\dots,Y^{n,\ell}_{\ell}$ and $Z^{n,\ell}_0,\dots,Z^{n,\ell}_{\ell}$ with joint distribution given by 
\begin{center}
\resizebox{\textwidth}{!}{$
p\left(Z^{n,\ell}_0=c_0,\dots,Z^{n,\ell}_{\ell}=c_\ell , X^{n,\ell}_0=a_0,\dots,X^{n,\ell}_{\ell}=a_\ell, Y^{n,\ell}_0=b_0,\dots,Y^{n,\ell}_{\ell}=b_\ell, C=d\right)
\coloneqq\frac{S(w,v,u)}{n-\ell},
$}
\end{center}
where $S(w,v,u)$ is  the number of indexes $i$ such that  $i<n-\ell$, $w[i:i+\ell]=c_0\dots c_\ell$, $u[i:i+\ell]=a_0\dots a_\ell$, $v[i:i+\ell]=b_0\dots b_\ell$ and 
$c_0\dots c_\ell$ is the binary representation of the  sum of the numbers represented by $a_0\dots a_\ell$ and $b_0\dots b_\ell$ with carry bit equal to~$d$.

For a short  notation, let
$\bar X^{n,\ell}=X^{n,\ell}_0,\dots,X^{n,\ell}_{\ell}$, 
$\bar Y^{n,\ell}=Y^{n,\ell}_0,\dots,Y^{n,\ell}_{\ell}$, and
$\bar Z^{n,\ell}=Z^{n,\ell}_0,\dots,Z^{n,\ell}_{\ell}$.
Since $\bar Z^{n,\ell}$ is completely determined by $\bar X^{n,\ell}$, $\bar Y^{n,\ell}$ and $C$,
we know that 
\[
H(\bar Z^{n,\ell}, \bar X^{n,\ell}, \bar Y^{n,\ell}, C)=H(\bar X^{n,\ell}, \bar Y^{n,\ell}, C)
\]
and by the chain rule we know that 
\[
H(\bar Z^{n,\ell}, \bar X^{n,\ell}, \bar Y^{n,\ell}, C)=H(\bar Z^{n,\ell})+H(\bar X^{n,\ell}, \bar Y^{n,\ell}, C\mid \bar Z^{n,\ell}).
\]
Then 
\[
H(\bar X^{n,\ell}, \bar Y^{n,\ell}, C)=H(\bar Z^{n,\ell})+H(\bar X^{n,\ell}, \bar Y^{n,\ell}, C\mid \bar Z^{n,\ell})
\]
and since $H(\bar X^{n,\ell}, \bar Y^{n,\ell}, C\mid \bar Z^{n,\ell})⩾0$ we conclude 
\[
H(\bar Z^{n,\ell})⩽H(\bar X^{n,\ell}, \bar Y^{n,\ell}, C).
\]
% Since $\bar X^{n,\ell}$ is completely determined by $\bar Y^{n,\ell}$, $\bar Z^{n,\ell}$ and $C$,
By the subadditivity of $H$,
$$
H(\bar Z^{n,\ell})⩽H(\bar X^{n,\ell}) + H(\bar Y^{n,\ell}) + H(C).
$$
Then, after dividing by $\ell+1$ and taking $\liminf$ we obtain
\begin{equation}\label{eqn:liminfH}
\liminf_{n\to\infty}H(\bar Z^{n,\ell})/(\ell+1)⩽\liminf_{n\to\infty}H(\bar X^{n,\ell})/(\ell+1) + \limsup_{n\to\infty}H(\bar Y^{n,\ell})/(\ell+1) + H(C)/(\ell+1),
\end{equation}
which, by the equivalence of the non-aligned block entropy and $H$ given in~\eqref{unh:H} and~\eqref{onh:H}, is equivalent to $\unhl(z)⩽\unhl(x)+\onhl(y)$. Since for any $x$ the limit of $\unh_\ell(x)$ 
and the limit of $\onh_\ell(x)$ exist when $\ell$ goes to infinity, we conclude that $\unh(z)⩽\unh(x) + \onh(y)$. This shows item 1 of Theorem~\ref{thm:subadditive}. For item~2, we reason analogously except that~\eqref{eqn:liminfH} becomes
$$
\limsup_{n\to\infty}H(\bar Z^{n,\ell})/(\ell+1)⩽\limsup_{n\to\infty}H(\bar X^{n,\ell})/(\ell+1) + \limsup_{n\to\infty}H(\bar Y^{n,\ell})/(\ell+1) + H(C)/(\ell+1).
$$
and this concludes the proof.
\end{proof}

\section{Proof of Theorem~\ref{thm:main}} \label{sec:main}

By the definition of the joint probability 
in~\eqref{eqn:distribution} we have,
\begin{align*}
\probi{n}{\ell}{0}{\ell-1}{\ell-1}&=\frac{\occ{w[0:|w|-1)}{a_0\dots a_{\ell-1}}}{n-\ell}
\\
\probi{n}{\ell}{0}{\ell}{\ell}&=\frac{\occ{w}{a_0\dots a_\ell}}{n-\ell}
\\
p(X^{n,\ell}_\ell=a_\ell\mid X^{n,\ell}_0=a_0,\dots,X^{n,\ell}_{\ell-1}=a_{\ell-1})&=\frac{\occ{w}{\bar aa_\ell}}{\occ{w[0:|w|-1)}{a_0\dots a_{\ell-1}}}.
\end{align*}
For ease of notation we write $\bar a=a_0\dots a_{\ell-1}$, and for the
joint probability and for the conditional probability we write
% \sideolivier{I would write $p(\bar a, a_\ell)$ rather than $p(a_\ell,\bar a)$
%   because $a_\ell$ comes after.  Or course I would still write
%   $p(a_\ell\mid \bar a)$.}
% \sidevero{ok with me.}

\begin{align*}
p(\bar a)&=\probi{n}{\ell}{0}{\ell-1}{\ell-1}
\\
p(\bar a,a_\ell)&=\probi{n}{\ell}{0}{\ell}{\ell}
\\
p(a_\ell\mid \bar a)&=p(X^{n,\ell}_\ell=a_\ell\mid X^{n,\ell}_0=a_0,\dots,X^{n,\ell}_{\ell-1}=a_{\ell-1}).
\end{align*}
Then we have
\begin{align*}
\hc{n}{\ell}{\ell}{0}&= -\sum_{\bar a=a_0,\dots,a_{\ell-1}}p(\bar a)\sum_{a_\ell} p(a_\ell\mid \bar a)\log p(a_\ell\mid \bar a)  
\\
&=-\sum_{\bar a=a_0,\dots,a_{\ell-1}}p(\bar a) \left( p(0\mid \bar a)\log p(0\mid \bar a) + p(1\mid \bar a)\log p(1\mid \bar a) \right) 
\\
&=-\sum_{\bar a=a_0,\dots,a_{\ell-1}}p(\bar a) \left( (1-p(1\mid \bar a))\log (1-p(1\mid \bar a)) + p(1\mid \bar a)\log p(1\mid \bar a) \right) 
\\
&= \sum_{\bar a=a_0,\dots,a_{\ell-1}} p(\bar a)\cdot \entropyprob(p(1\mid \bar a)),
\end{align*}
where each $a_i$ in the range of the sums are elements of $\B$.
The {\em probability of error} of $\bar a=a_0,\dots,a_{\ell-1}$ is defined by: 
$$
e(\bar a)\coloneqq\min(p(0\mid \bar a),p(1\mid \bar a))=\min(p(1\mid \bar a),1-p(1\mid \bar a)).
$$

On the one hand, since $\entropyprob(\alpha)=\entropyprob(1-\alpha)$ we have
$$
\hc{n}{\ell}{\ell}{0}=\sum_{\bar a=a_0,\dots,a_{\ell-1}} p(\bar a)\cdot \entropyprob(e(\bar a)).
$$
On the other, one can show that for $w=x[0:n)$ we have
$$
\gammarl(w)= \sum_{\bar a=a_0,\dots,a_{\ell-1}} p(\bar a)\cdot e(\bar a).
$$

\begin{proof}[Proof of item~1 of Theorem~\ref{thm:main}]
Since  $2p⩽\entropyprob(p)$ for $p\in[0,1/2]$ and $e(\bar a)\in[0,1/2]$, we conclude 
$$
2\gammarl(w)⩽\hc{n}{\ell}{\ell}{0}.
$$

By Jensen inequality applied to $\entropyprob$ we have 
\begin{align*}
\hc{n}{\ell}{\ell}{0}&=\sum_{\bar a=a_0,\dots,a_{\ell-1}} p(\bar a)\cdot \entropyprob(e(\bar a))
\\
&⩽\entropyprob\Big(\sum_{\bar a=a_0,\dots,a_{\ell-1}} p(\bar a)\cdot e(\bar a)\Big)=\entropyprob(\gammarl(w)).
\end{align*}
This shows that 
   $2\ugammar(x)  ⩽ \uhc(x)  ⩽ \entropyprob(\ugammar(x))$ 
and
   $2\ogammar(x) ⩽ \ohc(x)  ⩽ \entropyprob(\ogammar(x))$.
 The inequalities\linebreak 
   $2\ubetar(x) ⩽ \uhc(x)  ⩽ \entropyprob(\ubetar(x))$
and
   $2\obetar(x) ⩽ \ohc(x)  ⩽ \entropyprob(\obetar(x))$ 
can be proved similarly  changing~\eqref{eqn:distribution}~to
\begin{equation*}
\probi{n}{\ell}{0}{\ell}{\ell}=
\frac{\occ{w}{a_{\ell}a_0\cdots a_{\ell-1}}}{n-\ell}.
\end{equation*}
and this concludes the proof.
\end{proof}

The following proposition proves  item 2 of Theorem~\ref{thm:main},
 which asserts that  the inequalities of item~1 
 are sharp.
\begin{proposition}\label{prop:sharp}
  For every  $α$ and $ε$ such that $0 ⩽ α ⩽ 1/2$ and $ε >0$
  there exist  $x,x'∈\B^ℕ$ such that
  \[
  \ugammar(x) = \ogammar(x) = \ubetar(x) = \obetar(x) = \ugammar(x') =
  \ogammar(x') = \ubetar(x') = \obetar(x') = α,\]
  \[ \unh(x) = \onh(x) ⩽ 2α + ε\qquad 
\text{  and }\qquad
  \unh(x') = \onh(x') = \entropyprob(α).
  \]
% \note{S: second line said $2α + ε\unh(x) = \onh(x)$ in the submission. Corrected}  
\end{proposition}

Before giving the proof, we define the notion of genericity. Consider a Markov chain with state space $Q$, transitions given by stochastic matrices $P^{(0)}$ and $P^{(1)}$ corresponding to symbols 0 and 1 respectively, and stationary distribution $\pi$. We define the measure $\mu$ on cylinder sets of $\B^\N$ as follows:
\[
\mu([w]) = \sum_{q_1, q_2, \dots, q_{|w|} \in Q} \pi(q_1) \prod_{j=1}^{|w|-1} P^{(w(j))}_{q_j q_{j+1}}
\]
where $[w]$ is the cylinder of $w\in\B^*$.

A binary sequence $x\in\B^\N$ is called {\em generic} for $\mu$, if for every $w\in\B^*$, the  frequency of $w$ in $x$ converges to the probability assigned to the cylinder set $[w]$ by the measure $\mu$:
\[
\lim_{n \to \infty} \frac{\occ{x[0:n)}{w}}{n} = \mu([w]).
\]
It is a well known fact in ergodic theory almost all sequences are generic.

\begin{proof}[Proof of Proposition \ref{prop:sharp}]
  We start with the definition of the sequence~$x'$.  Let $μ$ be the
  Bernoulli measure given by~$μ(0) = α$ and $μ(1) = 1-α$ and let $x'$ be a
  generic sequence for~$μ$.  
  It is straightforward to verify all equalities
  involving~$x'$.

  The construction of the sequence~$x$ is a bit more involved.  Let $p$
  and~$q$ be two integers and $α'$ be a real number such that
  $α = p/2q + α'/q$.  Choosing $p$ and~$q$ large enough can guarantee that
  $\entropyprob(α')/q ⩽ 1/q < ε$. Let $y$ be a normal sequence over~$\B$ and
  let $z$ be a generic sequence for the Bernoulli measure~$μ'$ given by
  $μ'(0) = α'$ and $μ'(1) = 1-α'$.  The sequence $x = (x_n)_{n ⩾ 0}$ is
  then defined by
  \begin{displaymath}
    x_n =
    \begin{cases}
      y_n & \text{if $0 ⩽ n \bmod q ⩽ p-1$}, \\
      z_n & \text{if $n \bmod q = p$}, \\
      0 & \text{if $p+1 ⩽ n \bmod q ⩽ q-1$}. 
    \end{cases}
  \end{displaymath}
  It can be verified that
  $\ugammar(x) = \ogammar(x) = \ubetar(x) = \obetar(x) = α$ and that
  $\unh((x) = \onh(x) = p/q + \entropyprob(α')/q⩽ 2α + ε$.
  This completes the proof of the proposition.
\end{proof}

\section{Proof of Proposition~\ref{prop:gamma-beta}}
\label{sec:gamma-beta}

We show that the two functions $\ubetar$ and $\ugammar$ are different, as are $\obetar$ and $\ogammar$.
Precisely, we provide a sequence~$x ∈
\B^ℕ$ such that $\ubetar(x) = \obetar(x) ≠ \ugammar(x) = \ogammar(x)$.  In
the rest of this section, we write $\betar(x)$ and $\gammar(x)$ for
$\ubetar(x) = \obetar(x)$ and $\ugammar(x) = \ogammar(x)$ respectively.
This sequence $x$ is obtained as a generic sequence for a Markov chain.  We need two  auxiliary results, Lemmas~\ref{lem:snake} and~\ref{lem:freqs2betagamma} below.

\begin{lemma}\label{lem:snake}
  Let $x$ be a generic sequence of a Markov chain with stationary
  distribution~$π$.  Let $θ_{i,b}$ be the probability of having symbol~$b$
  after state~$i$. Then
  \begin{equation}\label{conjecture}
    \gammar(x) ⩾ ∑_{i ∈ Q} π_i\min(θ_{i,0}, θ_{i,1}).
  \end{equation}
\end{lemma}

Before giving the proof of Lemma~\ref{lem:snake}, we give an informal but
more intuitive proof.  Suppose that the prediction of the next symbol is
based of the last $\ell$ symbols.  The prediction can only be better if the
last $\ell$ states are added to the last $\ell$ symbols.  The main property
of Markov chains implies that the last state is sufficient.  This proves
that $\gammarl(x) ⩾ ∑_{i ∈ Q} π_i\min(θ_{i,0}, θ_{i,1})$ for each $ℓ ⩾ 1$
and thus $\gammar(x) ⩾ ∑_{i ∈ Q} π_i\min(θ_{i,0}, θ_{i,1})$.

The proof of the lemma is based on the notion of the \emph{snake
  Markov Chain} \cite[Problem~2.2.4]{Bremaud08} that we now introduce.  Let
a Markov chain with state space~$Q$ and transition matrix~$P$.  Let $k ⩾ 1$
be a positive integer.  

The \emph{snake} Markov chain of order~$k$ is the Markov
chain whose state space is the set $\{ q_1,…,q_k ∈ Q^k: P_{q_1,q_2} ⋯
P_{q_{k-1},q_k} ≠ 0\}$ of sequences of $k$ states in the original Markov
chain.  Each such sequence of~$k$ states can be viewed as a path made of
$k-1$ consecutive transitions. There is a transition from state $q_1,…,q_k$
to state $q'_1,…,q'_k$ if $q'_i = q_{i+1}$ for each $1 ⩽ i ⩽ k-1$ and its
probability is $P_{q'_{k-1},q'_k} = P_{q_k,q'_k}$.  If the original Markov
chain is irreducible, so is the snake Markov chain. Furthermore, if $π =
(π_q)_{q ∈ Q}$ is the stationary distribution of the original Markov chain,
the stationary distribution of the snake Markov chain is given by
$π_{q_1,…,q_k} = π_{q_1} P_{q_1,q_2} ⋯ P_{q_{k-1},q_k}$.  Note that the
{\em snake} Markov chain of order~$1$ is identical to the starting Markov chain.

\begin{proof}[Proof of Lemma~\ref{lem:snake}]
  Let us fix a positive integer~$ℓ ⩾ 1$ and a function
  $f : \B^ℓ → \B$.  Let us consider the snake Markov chain
  of order~$ℓ+1$.  Each state of this latter Markov chain is a sequence of
  $ℓ$ consecutive transitions forming a path of length~$ℓ$.  The label of
  this state is a word of length~$ℓ$ obtained by concatenating the labels
  of its $ℓ$ transitions.  The function~$f$ maps this word to a predicted
  symbol.  It is important to notice that the probabilities of transitions
  in the snake Markov chain are essentially the same as in the original
  Markov chain.  Since the sequence $x$ is generic  for the Markov chain, it
  is also generic  for the snake Markov chain.  Therefore for each state of
  the snake Markov chain, choosing the symbol with greatest probability
  instead of the symbol given by the function~$f$ decreases the error rate
  of prediction.  This proves the inequality.
\end{proof}
We conjecture that the inequality in~\eqref{conjecture}  in the statement of Lemma~\ref{lem:snake}
is, indeed, an equality.

\begin{lemma} \label{lem:freqs2betagamma}
  Let $ℓ$ be a positive integer and let $x$ be sequence such that the
  frequency in~$x$ of each word $u ∈ \B^{ℓ+1}$ does exist and is
  equal to~$α_u$.  Then
  \begin{displaymath}
    \betarl(x) = ∑_{w ∈ \B^ℓ} \min(α_{0w},α_{1w})
    \qquad\text{and}\qquad
    \gammarl(x) = ∑_{w ∈ \B^ℓ} \min(α_{w0},α_{w1}).
  \end{displaymath}
\end{lemma}
\begin{proof}
  Note that the frequency of each word~$w$ of length~$ℓ$ in~$x$ is
  $α_{0w}+α_{1w} = α_{w0}+α_{w1}$ and that $α_{0w}/(α_{0w}+α_{1w})$
  (respectively $α_{1w}/(α_{0w}+α_{1w})$) is the probability of having
  a~$0$ (respectively a~$1$) before~$w$.  It follows that $\betarl$ can be
  written
  \begin{displaymath}
    \betarl(x) = ∑_{w ∈ \B^ℓ} (α_{0w}+α_{1w})
    \min\left(\frac{α_{0w}}{α_{0w}+α_{1w}},
              \frac{α_{1w}}{α_{0w}+α_{1w}}\right).
  \end{displaymath}
and this concludes the proof.
\end{proof}

\begin{figure}[htbp]
  \begin{center}
    \begin{tikzpicture}[>=stealth',initial text=,semithick,auto,inner sep=2pt]
      \begin{scope}
      \node[style={circle,draw}] (0) at (0,2) {$0$};
      \node[style={circle,draw}] (1) at (2,2) {$1$};
      \node[style={circle,draw}] (2) at (2,0) {$2$};
      \node[style={circle,draw}] (3) at (0,0) {$3$};
      \path[->] (0) edge[bend left=15] node {$0{:}\tfrac{1}{2}$} (1);
      \path[->] (0) edge[bend left=15] node {$1{:}\tfrac{1}{2}$} (3);
      \path[->] (1) edge[bend left=15] node {$0{:}\tfrac{1}{2}$} (2);
      \path[->] (1) edge[bend left=15] node {$1{:}\tfrac{1}{2}$} (0);
      \path[->] (2) edge[bend left=15] node {$0{:}\tfrac{1}{2}$} (3);
      \path[->] (2) edge[bend left=15] node {$1{:}\tfrac{1}{2}$} (1);
      \path[->] (3) edge[bend left=15] node {$0{:}\tfrac{1}{3}$} (0);
      \path[->] (3) edge[bend left=15] node {$1{:}\tfrac{2}{3}$} (2);
      \end{scope}
    \end{tikzpicture}
  \end{center}
  \caption{The Markov chain}
  \label{fig:Markovchain}
\end{figure}

\begin{proof}[Proof of Proposition~\ref{prop:gamma-beta}] 
We construct a generic sequence $x$ for the following  Markov chain.
%We first describe the Markov chain. 
Its state space is $Q = \{0,1,2,3\}$
and its probability matrix is this 
 stochastic matrix,
\begin{displaymath}
  P \coloneqq
  \left(
    \begin{matrix}
      0 & \tfrac{1}{2} & 0 & \tfrac{1}{2} \\
      \tfrac{1}{2} & 0 & \tfrac{1}{2} & 0 \\
      0 & \tfrac{1}{2} & 0 & \tfrac{1}{2} \\
      \tfrac{1}{3} & 0 & \tfrac{2}{3} & 0 
    \end{matrix}
  \right)
\end{displaymath}
Each transition with a non-zero probability is labeled with a symbol $b ∈
\B$ in such a way that different transitions leaving the same state
have different labels.  Therefore the Markov chain can be viewed as a
deterministic finite-state automaton with probabilistic transitions.  This automaton is
pictured in Figure~\ref{fig:Markovchain}.  A transition of the form
$p \trans{b:α} q$ means that the probability of the transition from~$p$
to~$q$ is~$α ∈ [0,1]$ and that its label is $b ∈ \B$.
The stationary distribution~$π$ of the Markov chain is the vector
$π = [5/24, 1/4, 7/24, 1/4]$.

Lemma~\ref{lem:snake} allows us to give the lower bound~$11/24$
for $\gammar(x)$.  Lemma~\ref{lem:freqs2betagamma} allows us to compute
$\betarl$ and $\gammarl$ from the frequencies in~$x$ of words of
length~$ℓ+1$.  From the frequencies of words of length~$7$ in~$x$, we get
that $\betar_6(x) = 9503/20736$ and thus $\betar(x) < 9503/20736 < 11/24 ⩽
\gammar(x)$.
\end{proof}

\section*{Acknowledgments}

The authors are very grateful to an anonymous referee for having suggested that the proof could be presented exclusively based on conditional entropy, making the proofs much shorter. The authors also thank to the editor  that handled this article in Information and Computation.
The three authors are members of \href{http://www.irp-sinfin.org/}{IRP SINFIN}, Université Paris
Cité/CNRS-Universidad de Buenos Aires-CONICET.  This project has been
developed with Argentine grants 
PICT-2021-I-A-00838,
CONICET PIP 11220210100220CO and 
UBACyT 20020220100065BA, UBACyT 20020190100021BA.  The second author is supported by ANR SymDynAr
(ANR-23-CE40-0024-01).

\bibliographystyle{plain}
\bibliography{rauzydim}
\medskip

\noindent
Verónica Becher
\\
Departamento de Computación, Facultad de Ciencias Exactas y Naturales, Universidad de Buenos Aires  e ICC CONICET\\
Pabellón 0,  Ciudad Universitaria,  C1428EGA Ciudad Autónoma de Buenos Aires, Argentina\\
{\tt vbecher@dc.uba.ar}
\medskip

\noindent
Olivier Carton \\
Institut de Recherche en Informatique Fondamentale \\
Universit\'e Paris Diderot, France \\
{\tt Olivier.Carton@irif.fr}
\medskip

\noindent
Santiago Figueira
\\
Departamento de Computación, Facultad de Ciencias Exactas y Naturales, Universidad de Buenos Aires  e ICC CONICET\\
Pabellón 0,  Ciudad Universitaria,  C1428EGA Ciudad Autónoma de Buenos Aires, Argentina\\
{\tt santiago@dc.uba.ar}
\bigskip

\end{document}